\documentclass[smallextended,natbib]{svmult}
\usepackage{latexsym}
\usepackage{graphicx}
\usepackage{amsmath}
\usepackage{mathptmx}
\usepackage{pdfpages}
\usepackage{array}
\usepackage[center]{caption}
\usepackage{algorithm}
\usepackage{algpseudocode}
\usepackage[figuresright]{rotating}
\usepackage{booktabs}
\usepackage{color, colortbl}
\usepackage{array}
\usepackage{hhline}

\begin{document}

\title*{Heuristic for Maximizing Grouping Efficiency in the Cell Formation Problem}

\author{Ilya Bychkov \and Mikhail Batsyn \and Panos M. Pardalos}

\institute
{
  Ilya Bychkov \and Mikhail Batsyn \and Panos M. Pardalos \at Laboratory of Algorithms and Technologies for Network Analysis, National Research University Higher School of Economics, 136 Rodionova Street, Nizhny Novgorod, Russia, 603093, \email{ibychkov@hse.ru, mbatsyn@hse.ru, pardalos@ufl.edu}
\and
Panos M. Pardalos, \at Center for Applied Optimization, University of Florida, USA, 401 Weil Hall, P.O. Box 116595, Gainesville, FL 32611-6595, \email{pardalos@ufl.edu}
}
\maketitle

\abstract
{
In our paper we consider the Cell Formation Problem in Group Technology with grouping efficiency as an objective function. We present a heuristic approach for obtaining high-quality solutions of the CFP. The suggested heuristic applies an improvement procedure to obtain solutions with high grouping efficiency. This procedure is repeated many times for randomly generated cell configurations. Our computational experiments are performed for popular benchmark instances taken from the literature with sizes from 10x20 to 50x150. Better solutions unknown before are found for 23 instances of the 24 considered.
}

\section{Introduction}

\citet{16} was the first who formulated the main ideas of the group 
technology. The notion of the Group Technology was introduced in Russia by 
\citet{34}, though his work was translated to English only in 1966 \citep{35}. 
One of the main problems stated by the Group Technology is the optimal 
formation of manufacturing cells, i.e. grouping of machines and parts into 
cells such that for every machine in a cell the number of the parts from 
this cell processed by this machine is maximized and the number of the parts 
from other cells processed by this machine is minimized. In other words the 
intra-cell loading of machines is maximized and simultaneously the 
inter-cell movement of parts is minimized. This problem is called the Cell 
Formation Problem (CFP). \citet{6} suggested his Product Flow Analysis 
(PFA) approach for the CFP, and later popularized the Group Technology and 
the CFP in his book \citep{7}. 

The CFP is NP-hard since it can be reduced to the clustering problem \citep{3}. 
That is why there is a great 
number of heuristic approaches for solving CFP and almost no exact ones. The 
first algorithms for solving the CFP were different clustering techniques. 
Array-based clustering methods find rows and columns permutations of the 
machine-part matrix in order to form a block-diagonal structure. These 
methods include: Bond Energy Algorithm (BEA) of \citet{33}, 
Rank Order Clustering (ROC) algorithm by \citet{20}, its improved version 
ROC2 by \citet{21}, Direct Clustering Algorithm (DCA) of 
\citet{9}, Modified Rank Order Clustering (MODROC) algorithm by 
\citet{11}, the Close Neighbor Algorithm (CAN) 
by \citet{5}. Hierarchical clustering methods at first form 
several big cells, then divide each cell into smaller ones and so on 
gradually improving the value of the objective function. The most well-known 
methods are Single Linkage \citep{32}, Average Linkage 
\citep{44} and Complete Linkage \citep{36} algorithms. 
Non-hierarchical clustering methods are iterative approaches which start 
from some initial partition and improve it iteratively. The two most 
successful are GRAFICS algorithm by \citet{47} and 
ZODIAC algorithm by \citet{12}. A number of 
works considered the CFP as a graph partitioning problem where machines are 
vertices of a graph. \citet{42} used clique partitioning
of the machines graph. \citet{1} implemented a heuristic 
partitioning algorithm to solve CFP. \citet{39} and \citet{40} suggested an 
algorithm based on the minimum spanning tree problem. Mathematical
programming approaches are also very popular for the CFP. Since the 
objective function of the CFP is rather complicated from the mathematical 
programming point of view most of the researchers use some approximation 
model which is then solved exactly for small instances and heuristically for
large. \citet{27} formulated CFP via p-median model and solved several 
small-size CFP instances, \citet{45} used Generalized Assignment Problem 
as an approximation model, \citet{50} proposed a simplified p-median 
model for solving large CFP instances, \citet{22} 
applied minimum k-cut problem to the CFP, 
\citet{23} used p-median approximation model and solved it exactly by means of 
their pseudo-boolean approach including large CFP instances up to 50x150 
instance. A number of meta-heuristics have been applied recently to the CFP. 
Most of these approaches can be related to genetic, simulated annealing, 
tabu search, and neural networks algorithms. Among them such works as: 
\citet{17}, \citet{51}, \citet{52}, \citet{30}, \citet{31}, \citet{53}.

Our heuristic algorithm is based on sequential improvements of the solution.
We modify the cell configuration by enlarging one cell and reducing another.
The basic procedure of the algorithm has the following steps:
\begin{enumerate}
\item Generate a random cell configuration.
\item Improve the initial solution moving one row or column from one cell to another until the grouping efficiency is increasing.
\item Repeat steps 1-2 a predefined number of times (we use 2000 times for computational experiments in this paper).
\end{enumerate}

The paper is organized as follows.
In the next section we provide the Cell Formation Problem formulation.
In section \ref{algorithm} we present our improvement heuristic that allows us 
to get good solutions by iterative modifications of cells which lead to increasing of the objective function.
In section \ref{computres} we report our computational results 
and section \ref{conclusion} concludes the paper with a short summary.

\section{The Cell Formation Problem}
\label{cfpform}
The CFP consists in an optimal grouping of the given machines and parts into 
cells. The input for this problem is given by $m$ machines, $p$ parts and a 
rectangular machine-part incidence matrix $A=[a_{ij} ]$, where $a_{ij} =1$ 
if part $j$ is processed on machine $i$. The objective is to find an optimal 
number and configuration of rectangular cells (diagonal blocks in the 
machine-part matrix) and optimal grouping of rows (machines) and columns 
(parts) into these cells such that the number of zeros inside the 
chosen cells (voids) and the number of ones outside these cells (exceptions) 
are minimized. A concrete combination of rectangular cells in a solution (diagonal blocks in the 
machine-part matrix) we will call a cells configuration. Since it is usually not possible to minimize these two 
values simultaneously there have appeared a number of compound criteria 
trying to join it into one objective function. Some of them are presented 
below. 

For example, we are given the machine-part matrix \citep{49} shown in table \ref{table1}.
Two different solutions for this CFP are shown in tables \ref{table2} and \ref{table3}.
The left solution is better because it has less voids (3 against 4) and 
exceptions (4 against 5) than the right one. But one of its cells is a 
singleton -- a cell which has less than two machines or parts. In some 
CFP formulations singletons are not allowed, so in this case this solution 
is not feasible. In this paper we consider the both cases (with allowed 
singletons and with not allowed) and when there is a solution with 
singletons found by the suggested heuristic better than without singletons 
we present the both solutions.

\begin{table}
\centering
        \begin{tabular}{ c|ccccccc|}
                  \multicolumn{1}{c}{~} & \multicolumn{1}{c}{$p_{1}$ } & \multicolumn{1}{c}{$p_{2}$ } & \multicolumn{1}{c}{$p_{3}$ } & \multicolumn{1}{c}{$p_{4}$ } & \multicolumn{1}{c}{$p_{5}$ } & \multicolumn{1}{c}{$p_{6}$ } & \multicolumn{1}{c}{$p_{7}$} \\ \hhline{~-------}
                  $m_{1}$ & \hspace{0.5em} 1\hspace{0.5em}  & \hspace{0.5em} 0\hspace{0.5em}  &\hspace{0.5em} 0\hspace{0.5em}  &\hspace{0.5em}0\hspace{0.5em}  &\hspace{0.5em}1\hspace{0.5em} &\hspace{0.5em}1\hspace{0.5em} &\hspace{0.5em}1\hspace{0.5em} \\ 
                  $m_{2}$ &0 &1 &1 &1 &1 &0 &0 \\ 
                  $m_{3}$ &0 &0 &1 &1 &1 &1 &0 \\ 
                  $m_{4}$ &1 &1 &1 &1 &0 &0 &0 \\ 
                  $m_{5}$ &0 &1 &0 &1 &1 &1 &0 \\ \hhline{~-------}
        \end{tabular}
        \caption{Machine-part 5x7 matrix from \citet{49}}
        \label{table1}
\end{table}

\begin{table}
\centering
        \begin{tabular}{ c|ccccccc|}
                  \multicolumn{1}{c}{~} & \multicolumn{1}{c}{$p_{7}$ } & \multicolumn{1}{c}{$p_{6}$ } & \multicolumn{1}{c}{$p_{1}$ } & \multicolumn{1}{c}{$p_{5}$ } & \multicolumn{1}{c}{$p_{3}$ } & \multicolumn{1}{c}{$p_{2}$ } & \multicolumn{1}{c}{$p_{4}$} \\ \hhline{~-------}
                  $m_{1}$ & \hspace{0.5em}\cellcolor{yellow} 1\hspace{0.5em}  & \cellcolor{yellow} 1\hspace{0.5em}  &\cellcolor{yellow} \hspace{0.5em}  1 \hspace{0.5em}  &\hspace{0.5em}1\hspace{0.5em}  &\hspace{0.5em}0\hspace{0.5em} &\hspace{0.5em}0\hspace{0.5em} &\hspace{0.5em}0\hspace{0.5em} \\ 
                  $m_{4}$ &0 &0 &1 &\cellcolor{yellow}0 &\cellcolor{yellow}1 &\cellcolor{yellow}1 &\cellcolor{yellow}1 \\ 
                  $m_{3}$ &0 &1 &0 &\cellcolor{yellow}1 &\cellcolor{yellow}1 &\cellcolor{yellow}0 &\cellcolor{yellow}1 \\ 
                  $m_{2}$ &0 &0 &0 &\cellcolor{yellow}1 &\cellcolor{yellow}1 &\cellcolor{yellow}1 &\cellcolor{yellow}1 \\ 
                  $m_{5}$ &0 &1 &0 &\cellcolor{yellow}1 &\cellcolor{yellow}0 &\cellcolor{yellow}1 &\cellcolor{yellow}1 \\ \hhline{~-------}
        \end{tabular}
        \caption{Solution with singletons}
        \label{table2}
\end{table}

\begin{table}
\centering
        \begin{tabular}{ c|ccccccc|}
                  \multicolumn{1}{c}{~} & \multicolumn{1}{c}{$p_{7}$ } & \multicolumn{1}{c}{$p_{1}$ } & \multicolumn{1}{c}{$p_{6}$ } & \multicolumn{1}{c}{$p_{5}$ } & \multicolumn{1}{c}{$p_{4}$ } & \multicolumn{1}{c}{$p_{3}$ } & \multicolumn{1}{c}{$p_{2}$} \\ \hhline{~-------}
                  $m_{1}$ &\hspace{0.5em}\cellcolor{yellow}1\hspace{0.5em}  & \cellcolor{yellow}\hspace{0.5em}1\hspace{0.5em}  &\hspace{0.5em} 1\hspace{0.5em}  &\hspace{0.5em}1\hspace{0.5em}  &\hspace{0.5em}0\hspace{0.5em} &\hspace{0.5em}0\hspace{0.5em} &\hspace{0.5em}0\hspace{0.5em} \\ 
                  $m_{4}$ &\cellcolor{yellow}0 &\cellcolor{yellow}1 &0 &0 &1 &1 &1 \\ 
                  $m_{2}$ &0 &0 &\cellcolor{yellow}0 &\cellcolor{yellow}1 &\cellcolor{yellow}1 &\cellcolor{yellow}1 &\cellcolor{yellow}1 \\ 
                  $m_{3}$ &0 &0 &\cellcolor{yellow}1 &\cellcolor{yellow}1 &\cellcolor{yellow}1 &\cellcolor{yellow}1 &\cellcolor{yellow}0 \\ 
                  $m_{5}$ &0 &0 &\cellcolor{yellow}1 &\cellcolor{yellow}1 &\cellcolor{yellow}1 &\cellcolor{yellow}0 &\cellcolor{yellow}1 \\ \hhline{~-------}
        \end{tabular}
        \caption{Solution without singletons}
        \label{table3}
\end{table}

There are a number of different objective functions used for the CFP. The 
following four functions are the most widely used:
\begin{enumerate}
\item Grouping efficiency suggested by \citet{13}:

\begin{equation}
\eta =q\eta _1 +(1-q)\eta _2 ,
\end{equation}
where
\[
\eta _1 =\frac{n_1 -n_1^{out} }{n_1 -n_1^{out} +n_0^{in} }=\frac{n_1^{in} 
}{n^{in}},
\]
\[
\eta _2 =\frac{mp-n_1 -n_0^{in} }{mp-n_1 -n_0^{in} +n_1^{out} 
}=\frac{n_0^{out} }{n^{out}},
\]
$\eta _1$ -- a ratio showing the intra-cell loading of machines (or the 
ratio of the number of ones in cells to the total number of elements in 
cells).

$\eta _2$ -- a ratio inverse to the inter-cell movement of parts (or the 
ratio of the number of zeroes out of cells to the total number of elements 
out of cells).

q -- a coefficient ($0\le q\le 1)$ reflecting the weights of the machine 
loading and the inter-cell movement in the objective function. It is usually 
taken equal to $\frac{1}{2}$, which means that it is equally important to maximize the 
machine loading and minimize the inter-cell movement.

$n_1$ -- a number of ones in the machine-part matrix, 

$n_0$ -- a number of zeroes in the machine-part matrix,

$n^{in}$ -- a number of elements inside the cells, 

$n^{out}$ -- a number of elements outside the cells,

$n_1^{in}$ -- a number of ones inside the cells, 

$n_1^{out} $ -- a number of ones outside the cells,

$n_0^{in}$ -- a number of zeroes inside the cells, 

$n_0^{out} $ -- a number of zeroes outside the cells.
\newline
\item Grouping efficacy suggested by \citet{24} :
\begin{equation}
\tau =\frac{n_1 -n_1^{out} }{n_1 +n_0^{in} }=\frac{n_1^{in} }{n_1 +n_0^{in} }
\end{equation}

\item Group Capability Index (GCI) suggested by \citet{18}:
\begin{equation}
GCI=1-\frac{n_1^{out} }{n_1 }=\frac{n_1 -n_1^{out} }{n_1 }
\end{equation}
\item Number of exceptions (ones outside cells) and voids (zeroes inside 
cells):
\begin{equation}
E+V=n_1^{out} +n_0^{in} 
\end{equation}
\end{enumerate}
The values of these objective functions for the solutions in tables \ref{table2} and \ref{table3} are shown below.
\[
\eta = \frac{1}{2} \cdot \frac{16}{19} + \frac{1}{2} \cdot \frac{12}{16} \approx 79.60\% \qquad
\eta = \frac{1}{2} \cdot \frac{15}{19} + \frac{1}{2} \cdot \frac{11}{16} \approx 73.85\%
\]
\[
\tau = \frac{20 - 4}{20 + 3} \approx 69.57\% \qquad
\tau = \frac{20 - 5}{20 + 4} \approx 62.50\%
\]
\[
GCI = \frac{20 - 4}{20} \approx 80.00\% \qquad
GCI = \frac{20 - 5}{20} \approx 75.00\%
\]
\[
E+V = 4 + 3 = 7 \qquad
E+V = 5 + 4 = 9
\]
In this paper we use the grouping efficiency measure and compare our computational results with the results of \citet{53} and \citet{23}.

The mathematical programming model of the CFP with the grouping efficiency objective function can be described using boolean variables $x_{ik}$ and $y_{jk}$.
Variable $x_{ik}$ takes value 1 if machine $i$ belongs to cell $k$ and takes value 0 otherwise.
Similarly variable $y_{jk}$ takes value 1 if part $j$ belongs to cell $k$ and takes value 0 otherwise. 
Machines index $i$ takes values from 1 to $m$ and parts index $j$ - from 1 to $p$. 
Cells index $k$ takes values from 1 to $c = \min(m, p)$ because every cell should contain at least one machine and one part and so the number of cells cannot be greater than $m$ and $p$.
Note that if a CFP solution has $n$ cells then for $k$ from $n+1$ to $c$ all variables $x_{ik}, y_{jk}$ will be zero in this model. So we can consider that the CFP solution always has $c$ cells, but some of them can be empty.
The mathematical programming formulation is as follows.
\begin{equation}
\max \left( \frac{n_1^{in}}{2n^{in}} + \frac{n_0^{out} }{2n^{out}} \right)
\end{equation}
where
\[
n^{in} = \sum_{k = 1}^{c} \sum_{i = 1}^{m} \sum_{j = 1}^{p} x_{ik} y_{jk}, \quad
n^{out} = mp - n^{in}
\]
\[
n_1^{in} = \sum_{k = 1}^{c} \sum_{i = 1}^{m} \sum_{j = 1}^{p} a_{ij} x_{ik} y_{jk}, \quad
n_0^{out} = n_0 - (n^{in} - n_1^{in})
\]
subject to
\begin{eqnarray}
\sum_{k = 1}^{c} x_{ik} = 1 \quad \forall i \in 1,...,m \\
\sum_{k = 1}^{c} y_{jk} = 1 \quad \forall j \in 1,...,p
\end{eqnarray}
\begin{eqnarray}
\sum_{i = 1}^{m} \sum_{j = 1}^{p} x_{ik}y_{jk} \ge \sum_{i = 1}^{m} x_{ik} \quad \forall k \in 1,...,c \\
\sum_{i = 1}^{m} \sum_{j = 1}^{p} x_{ik}y_{jk} \ge \sum_{j = 1}^{p} y_{jk} \quad \forall k \in 1,...,c 
\end{eqnarray}
\begin{eqnarray}
x_{ik} \in \{0, 1\} \quad \forall i \in 1,...,m \\
y_{jk} \in \{0, 1\} \quad \forall j \in 1,...,p
\end{eqnarray}
The objective function (5) is the grouping efficiency in this model.
Constraints (6) and (7) impose that every machine and every part belongs to some cell.
Constraints (8) and (9) guarantee that every non-empty cell contains at least one machine and one part.
Note that if singleton cells are not allowed then the right sides of inequalities (8) and (9) should have a coefficient of 2. All these constraints can be linearized in a standard way, but the objective function will still be fractional. That is why the exact solution of this problem presents considerable difficulties.

A cells configuration in the mathematical model is described by the number of machines $m_k$ and parts $p_k$ in every cell $k$.
\[
m_k = \sum_{i = 1}^{m} x_{ik}, \quad p_k = \sum_{j = 1}^{p} y_{jk}
\]
It is easy to see that when a cells configuration is fixed all the optimization criteria (1) - (4) become equivalent (proposition \ref{eqv}).
\begin{proposition}
\label{eqv}
If a cells configuration is fixed then objective functions (1) - (4): $\eta$, $\tau$, $GCI$, $E+V$ become equivalent and reach the optimal value on the same solutions. 
\end{proposition}
\begin{proof}
When a cells configuration is fixed the following values are constant: $m_k, p_k$.
The values of $n_1$ and $n_0$ are always constant. 
The values of $n^{in}$ and $n^{out}$ are constant since $n^{in} = \sum_{k=1}^c m_k p_k$ and $n^{out} = mp - n^{in}$. 
So if we maximize the number of ones inside the cells 
$n_1^{in} $ then simultaneously $n_0^{in} =n^{in}-n_1^{in} $ is minimized, $n_0^{out} =n_0 
-n_0^{in} $ is maximized, and $n_1^{out} =n_1 -n_1^{in} $ is minimized. This 
means that the grouping efficiency $\eta =q\frac{n_1^{in} 
}{n^{in}}+(1-q)\frac{n_0^{out} }{n^{out}}$ is maximized, the grouping 
efficacy $\tau =\frac{n_1^{in} }{n_1 +n_0^{in} }$ is maximized, the grouping 
capability index $GCI=1-\frac{n_1^{out} }{n_1 }$ is maximized, and the 
number of exceptions plus voids $E+V=n_1^{out} +n_0^{in} $ is minimized simultaneously on the same optimal solution. \qed
\end{proof}

\section{Algorithm description}
\label{algorithm}
The main function of our heuristic is presented by algorithm \ref{cmh-main}.
\begin{algorithm}
\caption{Main function}
\label{cmh-main}
\begin{algorithmic}
\Function {Solve}{ }
 \State \Call {FindOptimalCellRange} {$MinCells,MaxCells$} 
 \State $ConfigsNumber = 2000$ 
  \State $AllConfigs$ = \Call {GenerateConfigs}{$MinCells,MaxCells,ConfigsNumber$}
  \State \Return \Call {CMHeuristic}{$AllConfigs$} 
\EndFunction
\end{algorithmic}
\end{algorithm}
First we call \Call {FindOptimalCellRange} {$MinCells,MaxCells$} function that returns a potentially optimal range of cells - from $MinCells$ to $MaxCells$.
Then these values and $ConfigsNumber$ (the number of cell configurations to be generated) are passed to \Call {GenerateConfigs}{$MinCells,MaxCells,ConfigsNumber$} function which generates random cell configurations. The generated configurations $AllConfigs$ are passed to \Call {CMHeuristic}{$AllConfigs$}  function which finds a high-quality solution for every cell configuration and then chooses the solution with the greatest efficiency value.
\begin{algorithm}
\caption{Procedure for finding the optimal cell range}
\label{cmh-findcellrange}
\begin{algorithmic}
\Function {FindOptimalCellRange}{ $MinCells,MaxCells$}
  \If {($m > p$)}
   \State $minDimension = p$
   \Else
   \State $minDimension = m$
   \EndIf
   \State $ConfigsNumber = 500$
  \State $Configs$ = \Call {GenerateConfigs}{$2,minDimension,ConfigsNumber$}
  \State $Solution$ = \Call {CMHeuristic}{$Configs$}
  \State $BestCells$ = \Call {GetCellsNumber}{$Solution$}
  \State $MinCells$ = $BestCells$ - [$minDimension$ * 0,1 ] \Comment { [ ] - integer part }
  \State $MaxCells$ = $BestCells$ + [$minDimension$ * 0,1 ] 
\EndFunction
\end{algorithmic}
\end{algorithm}

In function \Call{FindOptimalCellRange}{ $MinCells,MaxCells$} (algorithm \ref{cmh-findcellrange}) we look over all the possible number of cells from 2 to maximal possible number of cells which is equal to $\min(m,p)$.
For every number of cells in this interval we generate a fixed number of configurations (we use 500 in this paper) calling \Call {GenerateConfigs}{$2,minDimension,ConfigsNumber$}
and then use our \Call {CMHeuristic}{$Configs$} to obtain a potentially optimal number of cells.
But we consider not only one number of cells but together with its 10\%-neighborhood $[MinCells,MaxCells]$.

\begin{algorithm}
\caption{Configurations generation}
\label{cmh-generate}
\begin{algorithmic}
\Function {GenerateConfigs}{$MinCells,MaxCells,ConfigsNumber$}
 \State $Configs = \emptyset $
 \For {$cells = MinCell,MaxCells $} 
 \State $Generated = \Call {GenerateConfigs}{cells,ConfigsNumber}$
 \State $Configs = Configs \cup Generated$
  \State \Return $Configs$
  \EndFor
\EndFunction
\end{algorithmic}
\end{algorithm}

Function \Call {GenerateConfigs}{$MinCells,MaxCells,ConfigsNumber$} (algorithm \ref{cmh-generate}) returns a set of randomly generated cell configurations with a number of cells ranging from $MinCells$ to $MaxCells$.
We call \Call {GenerateConfigsUniform}{$cells$, $ConfigsNumber$} function which randomly selects with uniform distribution $ConfigsNumber$ configurations from all possible cell configurations with the specified number of cells. 
Note that mathematically a cell configuration with $k$ cells can be represented as an integer partition of $m$ and $p$ values into sums of $k$ summands.
We form a set of configurations for every number of cells and then join them.

\begin{algorithm}
\caption {CMHeuristic}
\label{cmh}
\begin{algorithmic}
\Function {CMHeuristic}{$Configs$} 
\State $Best = 0$
\ForAll{$config \in Configs $}
\State $Solution = \Call {ImproveSolution}{config}$
\If {$Solution > Best$}
\State $Best = Solution$
\EndIf
\EndFor
\State \Return $Best$
\EndFunction
\end{algorithmic}
\end{algorithm}

Function \Call {CMHeuristic}{Configs} (algorithm \ref{cmh}) gets a set of cell configurations and
for each configuration runs an improvement algorithm to obtain a good solution. 
A solution includes a permuted machine-part matrix, a cell configuration, and the corresponding grouping efficiency value. 
The function chooses the best solution and returns it.

\begin{algorithm}
\caption {Solution improvement procedure}
\label{imp}
\begin{algorithmic}
\Function {ImproveSolution}{$config,\eta_{current}$} 
  \State $\eta_{current}$ = \Call {GroupingEfficiency}{config}
  \Repeat
  \State $PartFrom = 0$
  \State $PartTo = 0$
  \For{$part = 1,partsNumber$}
  \For{$cell = 1,cellsNumber$}
  \If {$(\eta_{part,cell} > \eta_{current})$}
   \State $\Delta_{parts} = (\eta_{part,cell} -\eta_{current})$
    \State $PartFrom = GetPartCell(part)$
    \State $PartTo = cell$    
  \EndIf
  \EndFor
  \EndFor

  \State $MachineFrom = 0$
  \State $MachineTo = 0$
  \For{$machine = 1,machinesNumber$}
  \For{$cell = 1,cellsNumber$}
  \If {$(\eta_{machine,cell} > \eta_{current})$}
      \State $\Delta_{machines} = (\eta_{machine,cell} -\eta_{current})$
      \State $MachineFrom$ = \Call {GetMachineCell}{$machine$}
      \State $MachineTo = cell$
  \EndIf
  \EndFor
  \EndFor
  \If {$\Delta_{parts} > \Delta_{machines}$} 
        \State \Call {MovePart}{$PartFrom,PartTo$}
   \Else 
        \State \Call {MoveMachine}{$MachineFrom,MachineTo$}    
  \EndIf
  \Until {$\Delta > 0$ } 
\EndFunction
\end{algorithmic}
\end{algorithm}

Improvement procedure \Call {ImproveSolution}{$config,\eta_{current}$} (algorithm \ref{imp}) works as follows.
We consider all the machines and the parts  in order to know if there is a machine or a part
that we can move to another cell and improve the current efficiency $\eta_{current}$.
First we consider moving of every part on all other cells and compute how the efficiency value changes.
Here $\eta_{part,cell}$ is the efficiency of the current solution where the part with index $part$ is moved to the cell with index $cell$.
This operation is performed for all the parts and the part with the maximum increase in efficiency $\Delta_{parts}$ is chosen.
Then we repeat the same operations for all the machines.
Finally, we compare the best part movement and the best machine movement and choose 
the one with the highest efficiency.
This procedure is performed until any improvement is possible and after that we get the final solution. 

The main idea of \Call {ImproveSolution}{$config,\eta_{current}$} is illustrated on \citet{44} instance 8x12 (table \ref{table6}).
\begin{table}[htbp]
  \centering
  \caption {\citep{44} instance 8x12}
    \begin{tabular}{p{0.5cm}>{\centering\arraybackslash}  p{0.5cm}>{\centering\arraybackslash}  p{0.5cm}>{\centering\arraybackslash}  p{0.5cm}>{\centering\arraybackslash}  p{0.5cm}>{\centering\arraybackslash}  p{0.5cm}>{\centering\arraybackslash}  p{0.5cm}>{\centering\arraybackslash}  p{0.5cm} >{\centering\arraybackslash} p{0.5cm}  >{\centering\arraybackslash}p{0.5cm}>{\centering\arraybackslash} p{0.5cm}>{\centering\arraybackslash} p{0.5cm} >{\centering\arraybackslash}p{0.5cm}>{\centering\arraybackslash} p{0.5cm}}
    \toprule
          & \textbf{1} & \textbf{2} & \textbf{3} & \textbf{4} & \textbf{5} & \textbf{6} & \textbf{7} & \textbf{8} & \textbf{9} & \textbf{10} & \textbf{11} & \textbf{12}   \\ 
    \midrule
    \textbf{1}   & \cellcolor{yellow}1    & \cellcolor{yellow}1     & \cellcolor{yellow}1     & 1     & 0     & 0     & 0     & 0     & 0     & 0     & 0     & 0   \\
    \textbf{2} &\cellcolor{yellow}1    &\cellcolor{yellow}0     &\cellcolor{yellow}1     & 1     & 1     & 1     & 1     & 0     & 0     & 1     & 0     & 0     \\
    \textbf{3} &\cellcolor{yellow}0     &\cellcolor{yellow}0    &\cellcolor{yellow}1     & 1     & 1     & 1     & 1     & 1     & 1     & 0     & 0     & 0     \\
    \textbf{4} & 0     & 0     & 0     &\cellcolor{yellow}0     &\cellcolor{yellow}0     &\cellcolor{yellow}1     &\cellcolor{yellow}1    &\cellcolor{yellow}1     &\cellcolor{yellow}1     & 1     & 0     & 0     \\
    \textbf{5} & 0     & 0     & 0     &\cellcolor{yellow} 0     &\cellcolor{yellow} 0     &\cellcolor{yellow} 0    &\cellcolor{yellow} 1    &\cellcolor{yellow} 1     &\cellcolor{yellow} 1     & 1     & 0     & 0     \\
    \textbf{6} & 0     & 0     & 0     &\cellcolor{yellow} 0     &\cellcolor{yellow} 0     &\cellcolor{yellow} 0     &\cellcolor{yellow} 1     &\cellcolor{yellow} 1     &\cellcolor{yellow} 1     & 0     & 1     & 0     \\
    \textbf{7} & 0     & 0     & 0     & 0     & 0     & 0     & 0     & 0     & 0     &\cellcolor{yellow} 0     &\cellcolor{yellow} 1     &\cellcolor{yellow} 1     \\
    \textbf{8} & 0     & 0     & 0     & 0     & 0     & 0     & 0     & 0     & 0     &\cellcolor{yellow} 0     &\cellcolor{yellow} 1     &\cellcolor{yellow} 1     \\
    \bottomrule
    \end{tabular}
  \label{table6}
\end{table}
To compute the grouping efficiency for this solution we need to know the number 
of ones inside cells $n_1 ^{in}$, the total number of elements  inside cells $n^{in} $, the 
number of zeros outside cells $n_0^{out} $ and the number of elements outside cells $n^{out}$.
The grouping efficiency is then calculated by the  following formula:
\[
\eta =q\cdot\frac{n_1^{in}}{n^{in}} +(1-q)\cdot\frac{n_0^{out} }{n^{out}} = \frac{1}{2}\cdot\frac{20}{33} +\frac{1}{2}\cdot\frac{48}{63} \approx 68.4\%  
\]

Looking at this solution (table \ref{table6}) we can conclude that it is possible for example to move part 4 from the second cell to the first one. And this way the number of zeros inside cells decreases by 3 and the number of ones outside cells also decreases by 4. So it is profitable to attach column 4 to the first cell as it is shown on table \ref{table7}.
\begin{table}[htbp]
  \centering
  \caption{Moving part 4 from cell 2 to cell 1}
    \begin{tabular}{>{\centering\arraybackslash}p{0.5cm} >{\centering\arraybackslash} p{0.5cm} >{\centering\arraybackslash} p{0.5cm}>{\centering\arraybackslash}  p{0.5cm}>{\centering\arraybackslash}  p{0.5cm}>{\centering\arraybackslash}  p{0.5cm}>{\centering\arraybackslash}  p{0.5cm}>{\centering\arraybackslash}  p{0.5cm}>{\centering\arraybackslash}  p{0.5cm}  >{\centering\arraybackslash}p{0.5cm} >{\centering\arraybackslash} p{0.5cm} >{\centering\arraybackslash} p{0.5cm} >{\centering\arraybackslash} p{0.5cm} >{\centering\arraybackslash} p{0.5cm}}
    \toprule
          & \textbf{1} & \textbf{2} & \textbf{3} & \textbf{4} & \textbf{5} & \textbf{6} & \textbf{7} & \textbf{8} & \textbf{9} & \textbf{10} & \textbf{11} & \textbf{12} \\ 
    \midrule
    \textbf{1}   & \cellcolor{yellow}1    & \cellcolor{yellow}1     & \cellcolor{yellow}1     & \cellcolor{yellow} 1     & 0     & 0     & 0     & 0     & 0     & 0     & 0     & 0     \\
    \textbf{2} &\cellcolor{yellow}1    &\cellcolor{yellow}0     &\cellcolor{yellow}1     & \cellcolor{yellow} 1     & 1     & 1     & 1     & 0     & 0     & 1     & 0     & 0     \\
    \textbf{3} &\cellcolor{yellow}0     &\cellcolor{yellow}0    &\cellcolor{yellow}1     & \cellcolor{yellow} 1     & 1     & 1     & 1     & 1     & 1     & 0     & 0     & 0     \\
    \textbf{4} & 0     & 0     & 0     &\cellcolor{red}0     &\cellcolor{yellow}0     &\cellcolor{yellow}1     &\cellcolor{yellow}1    &\cellcolor{yellow}1     &\cellcolor{yellow}1     & 1     & 0     & 0     \\
    \textbf{5} & 0     & 0     & 0     &\cellcolor{red} 0     &\cellcolor{yellow} 0     &\cellcolor{yellow} 0    &\cellcolor{yellow} 1    &\cellcolor{yellow} 1     &\cellcolor{yellow} 1     & 1     & 0     & 0    \\
    \textbf{6} & 0     & 0     & 0     &\cellcolor{red} 0     &\cellcolor{yellow} 0     &\cellcolor{yellow} 0     &\cellcolor{yellow} 1     &\cellcolor{yellow} 1     &\cellcolor{yellow} 1     & 0     & 1     & 0    \\
    \textbf{7} & 0     & 0     & 0     & 0     & 0     & 0     & 0     & 0     & 0     &\cellcolor{yellow} 0     &\cellcolor{yellow} 1     &\cellcolor{yellow} 1     \\
    \textbf{8} & 0     & 0     & 0     & 0     & 0     & 0     & 0     & 0     & 0     &\cellcolor{yellow} 0     &\cellcolor{yellow} 1     &\cellcolor{yellow} 1     \\
    \bottomrule
    \end{tabular}
  \label{table7}
\end{table}
For the modified cells configuration we have :
\[\eta =\frac{1}{2}\cdot\frac{23}{33} +\frac{1}{2}\cdot\frac{51}{63} \approx 75.32\%  \]
As a result the efficiency is increased almost for 7 percent. Computational results show that using such modifications could considerably improve the solution. The idea is to compute an increase in efficiency for each column and row when it is moved to another cell and then perform the modification corresponding to the maximal increase. For example, table \ref{table8} shows the maximal possible increase in efficiency for every row when it is moved to another cell.
\begin{table}[htbp]
  \centering
  \caption{Maximal efficiency increase for each row}
    \begin{tabular}{p{0.3cm}>{\centering\arraybackslash}  p{0.3cm}>{\centering\arraybackslash}  p{0.3cm}>{\centering\arraybackslash}  p{0.3cm}>{\centering\arraybackslash}  p{0.3cm}>{\centering\arraybackslash}  p{0.3cm}>{\centering\arraybackslash}  p{0.3cm}>{\centering\arraybackslash}  p{0.3cm} >{\centering\arraybackslash} p{0.3cm}  >{\centering\arraybackslash}p{0.3cm}>{\centering\arraybackslash} p{0.3cm}>{\centering\arraybackslash} p{0.3cm} >{\centering\arraybackslash}p{0.3cm}>{\centering\arraybackslash} p{0.3cm} p{1cm}}
    \toprule
          & \textbf{1} & \textbf{2} & \textbf{3} & \textbf{4} & \textbf{5} & \textbf{6} & \textbf{7} & \textbf{8} & \textbf{9} & \textbf{10} & \textbf{11} & \textbf{12} &       &  \\ 
    \midrule
    \textbf{1}   & \cellcolor{yellow}1    & \cellcolor{yellow}1     & \cellcolor{yellow}1     & 1     & 0     & 0     & 0     & 0     & 0     & 0     & 0     & 0     &       & - 6.94\% \\
    \textbf{2} &\cellcolor{yellow}1    &\cellcolor{yellow}0     &\cellcolor{yellow}1     & 1     & 1     & 1     & 1     & 0     & 0     & 1     & 0     & 0     &       & + 1.32\% \\
    \textbf{3} &\cellcolor{yellow}0     &\cellcolor{yellow}0    &\cellcolor{yellow}1     & 1     & 1     & 1     & 1     & 1     & 1     & 0     & 0     & 0     &       & + 7.99\% \\
    \textbf{4} & 0     & 0     & 0     &\cellcolor{yellow}0     &\cellcolor{yellow}0     &\cellcolor{yellow}1     &\cellcolor{yellow}1    &\cellcolor{yellow}1     &\cellcolor{yellow}1     & 1     & 0     & 0     &       & - 0.07\% \\
    \textbf{5} & 0     & 0     & 0     &\cellcolor{yellow} 0     &\cellcolor{yellow} 0     &\cellcolor{yellow} 0    &\cellcolor{yellow} 1    &\cellcolor{yellow} 1     &\cellcolor{yellow} 1     & 1     & 0     & 0     &       & + 0.77\% \\
    \textbf{6} & 0     & 0     & 0     &\cellcolor{yellow} 0     &\cellcolor{yellow} 0     &\cellcolor{yellow} 0     &\cellcolor{yellow} 1     &\cellcolor{yellow} 1     &\cellcolor{yellow} 1     & 0     & 1     & 0     &       & + 0.77\% \\
    \textbf{7} & 0     & 0     & 0     & 0     & 0     & 0     & 0     & 0     & 0     &\cellcolor{yellow} 0     &\cellcolor{yellow} 1     &\cellcolor{yellow} 1     &       & - 4.62\% \\
    \textbf{8} & 0     & 0     & 0     & 0     & 0     & 0     & 0     & 0     & 0     &\cellcolor{yellow} 0     &\cellcolor{yellow} 1     &\cellcolor{yellow} 1     &       & - 4.62\% \\
    \bottomrule
    \end{tabular}
  \label{table8}
\end{table}

\section{Computational results}
\label{computres}

\begin{sidewaystable}
\caption{Computational results}
\label{cmp}
\begin{tabular}{cccccccccc}

\hline
\raisebox{-3.00ex}[0cm][0cm]{{\#}}& 
\raisebox{-3.00ex}[0cm][0cm]{Source}& 
\raisebox{-3.00ex}[0cm][0cm]{mxp}& 
\multicolumn{5}{c}{Efficiency value,{\%}} & 
\raisebox{-3.00ex}[0cm][0cm]{Time, sec}& 
\raisebox{-3.00ex}[0cm][0cm]{Cells} \\
\cline{4-8} 
 & 
 & 
 & 
\raisebox{-1.50ex}[0cm][0cm]{Bhatnagar \& Saddikuti}& 
\raisebox{-1.50ex}[0cm][0cm]{Goldengorin et al.}& 
\multicolumn{3}{c}{Our} & 
 & 
 \\
\cline{6-8} 
 & 
 & 
 & 
 & 
 & 
Min& 
Avg& 
Max& 
 & 
 \\
\hline
1& 
\citet{54} & 
10x20& 
95.40& 
95.93 \footnote{This solution has a mistake.}& 
95.66& 
95.66& 
\textbf{95.66}& 
0.36& 
7 \\
\hline
2& 
\citet{55} & 
20x20& 
92.62& 
93.85& 
95.99& 
95.99& 
\textbf{95.99} & 
0.62& 
9 \\
\hline
3& 
\citet{38} & 
20x20& 
85.63& 
88.71& 
90,11& 
90,16& 
\textbf{90,22} & 
0,88& 
9 \\
\hline
4& 
\citet{5} & 
20x35& 
88.31& 
88.05& 
93,34& 
93,47& 
\textbf{93,55} & 
1,62& 
10 \\
\hline
5& 
\citet{8} & 
20x35& 
90.76& 
95.64& 
95,43& 
95,78& 
\textbf{95,79} & 
1,54& 
10 \\
\hline
6& 
\citet{55} & 
20x51& 
87.86& 
94.11& 
95,36& 
95,4& 
\textbf{95,45} & 
3,1& 
12 \\
\hline
7& 
\citet{13} & 
20x40& 
98.82& 
100.00& 
100& 
100& 
100& 
1,8& 
7 \\
\hline
8& 
\citet{13} & 
20x40& 
95.33& 
97.48& 
97,7& 
97,75& 
\textbf{97,76} & 
2,42& 
12 \\
\hline
9& 
\citet{13} &
20x40& 
93.78& 
96.36& 
96,84& 
96,88& 
\textbf{96,89} & 
2,56& 
12 \\
\hline
10& 
\citet{13} & 
20x40& 
87.92& 
94.32& 
96,11& 
96,16& 
\textbf{96,21} & 
3,3& 
15 \\
\hline
11& 
\citet{13} &
20x40& 
84.95& 
94.21& 
95,94& 
96,03& 
\textbf{96,1}& 
2,84& 
15 \\
\hline
12& 
\citet{13} &
20x40& 
85.06& 
92.32& 
95,85& 
95,9& 
\textbf{95,95}& 
2,76& 
15 \\
\hline
13& 
\citet{56} &
20x40& 
96.44& 
97.39& 
97,78& 
97,78& 
\textbf{97,78}& 
2,12& 
10 \\
\hline
14& 
\citet{56} &
20x40& 
92.35& 
95.74& 
97,4& 
97,4& 
\textbf{97,4}& 
2,2& 
14 \\
\hline
15& 
\citet{56} &
20x40& 
93.25& 
95.70& 
95,81& 
96,03& 
\textbf{96,17}& 
2,48& 
12 \\
\hline
16& 
\citet{56} &
20x40& 
91.11& 
96.40& 
96,98& 
96,98& 
\textbf{96,98}& 
2,78& 
14 \\
\hline
17& 
\citet{55} &
25x40& 
91.09& 
95.52& 
96,48& 
96,48& 
\textbf{96,48}& 
2,58& 
14 \\
\hline
18& 
\citet{53} &
28x35& 
93.43& 
93.82& 
94,81& 
94,85& 
\textbf{94,86}& 
2,46& 
10 \\
\hline
19& 
\citet{26} &
30x41& 
90.66& 
97.22& 
97,38& 
97,53& 
\textbf{97,62}& 
3,54& 
18 \\
\hline
20& 
\citet{48} &
30x50& 
88.17& 
96.48& 
96,77& 
96,83& 
\textbf{96,9}& 
5,02& 
18 \\
\hline
21& 
\citet{21} &
30x90& 
83.18& 
94.62& 
95,37& 
95,84& 
\textbf{96,27}& 
13,1& 
25 \\
\hline
22& 
\citet{12} &
40x100& 
94.75& 
95.91& 
98,06& 
98,1& 
\textbf{98,13}& 
16,88& 
17 \\
\hline
23& 
\citet{53} &
46x105& 
90.98& 
95.20& 
96,1& 
96,18& 
\textbf{96,29}& 
23,9& 
18 \\
\hline
24& 
\citet{31} &
50x150& 
93.05& 
92.92& 
96,08& 
96,17& 
\textbf{96,27}& 
51,66& 
24 \\
\hline
\end{tabular}
\end{sidewaystable}
In all the experiments for determining a potentially optimal range of cells we use 500 random cell configurations for each cells number and for obtaining the final solution we use 2000 random configurations. An Intel Core i7 machine with 2.20 GHz CPU and 8.00 Gb of memory is used in our experiments. We run our heuristic on 24 CFP benchmark instances taken from the literature. The sizes of the considered problems vary from 10x20 to 50x150. The computational results are presented in table \ref{cmp}. For every instance we make 50 algorithm runs and report minimum, average and maximum value of the grouping efficiency obtained by the suggested heuristic over these 50 runs. We compare our results with the best known values taken from \citet{23} and \citet{57}. We have found better solutions unknown before for 23 instances of the 24 considered. For CFP instance 6 we have found the same optimal solution with 100\% of grouping efficiency as in \citet{23}. For CFP instance 1 the solution of \citet{23} has some mistake. For this instance having a small size of 10x20 it can be proved that our solution is the global optimum applying an exact approach \citep{2} for the grouping efficiency objective and all the possible number of cells from 1 to 10.
\section{Concluding remarks}
\label{conclusion}
In this paper we present a new heuristic algorithm for solving the CFP.
The high quality of the solutions is achieved due to the enumeration of 
different numbers of cells and different cell 
configurations and applying our improvement procedure. Since the suggested heuristic works fast (the solution for 
one cell configuration is achieved in several milliseconds for any instance from 10x20 to 
50x150) we apply it for thousands of different configurations. Thus a big variety 
of good solutions is covered by the algorithm and the best of them has high 
grouping efficiency.

\section{Acknowledgements}
This work is partially supported by Laboratory of Algorithms and Technologies for Network Analysis, National Research University Higher School of Economics.\newline

\end{document}